\documentclass[journal,onecolumn]{IEEEtran}
\usepackage{graphicx} 
\usepackage{amsmath}
\usepackage{amssymb}
\usepackage{amsthm}
\newtheorem{theorem}{Theorem}
\newtheorem{remark}{Remark}

\newtheorem{assum}{Assumption}
\usepackage{multicol}
\usepackage{multirow}
\usepackage[english]{babel}   
\usepackage{color}
\usepackage{multicol}
\usepackage[T1]{fontenc}
\ifCLASSINFOpdf
\else
\fi
\begin{document}
%
\title{Adaptive-Robust Control of a Class of Nonlinear Systems with Unknown Input Delay}
%
%
%

\author{Spandan~Roy
        and~Indra~Narayan~Kar
\thanks{S. Roy and I. N. Kar are with the Department
of Electrical Engineering, Indian Institute of Technology-Delhi, New Delhi,
India e-mail: (sroy002@gmail.com, ink@ee.iitd.ac.in).}
\thanks{}}

%
%

\markboth{Arxiv-CS}%
{Shell \MakeLowercase{\textit{et al.}}: Bare Demo of IEEEtran.cls for Journals}
%



\maketitle

\begin{abstract}
In this paper, the tracking control problem of a class of uncertain Euler-Lagrange systems subjected to unknown input delay and bounded disturbances is addressed. To this front, a novel delay dependent control law, referred as Adaptive Robust Outer Loop Control (AROLC) is proposed. Compared to the conventional predictor based approaches, the proposed controller is capable of negotiating any input delay, within a stipulated range, without knowing the delay or its variation. The maximum allowable input delay is computed through Razumikhin-type stability analysis. AROLC also provides robustness against the disturbances due to input delay, parametric variations and unmodelled dynamics through switching control law. The novel adaptive law allows the switching gain to modify itself online in accordance with the tracking error without any prerequisite of the uncertainties. The uncertain system, employing AROLC, is shown to be Uniformly Ultimately Bounded (UUB). As a proof of concept, experimentation is carried out on a nonholonomic wheeled mobile robot with various time varying as well as fixed input delay, and better tracking accuracy of the proposed controller is noted compared to predictor based methodology. 
\end{abstract}

\begin{IEEEkeywords}
Adaptive-robust control, Euler-Lagrange systems, Input delay, Razumikhin theorem, Wheeled mobile robot.
\end{IEEEkeywords}

%
\IEEEpeerreviewmaketitle

\section{Introduction}
%
%
%
%
\subsection{Background and Motivation}
Time delays could be inherent in the system or may be a result of some sources such as transportation and transmission lags, communication delays etc. A few applications where delays are evident include, chemical and biological processes, teleoperated robotic systems, rolling mills, delays due to sensor response and input delay in mobile robots, engine cycle delays in internal combustion engine, control over network etc. \cite{Ref:1}-\cite{Ref:3}. Left unattended, such delay reduces performance of the system and may create potential instability to the system. The survey articles \cite{Ref:1}-\cite{Ref:5} document the recent advances, challenges and controllers developed to tackle the unwanted delay for linear systems \cite{Ref:1}-\cite{Ref:4} and nonlinear systems \cite{Ref:5}. The control strategies reported in literature can broadly be classified in two categories, one for linear systems (\cite{Ref:6}-\cite{Ref:30}) and other for nonlinear systems (\cite{Ref:20}-\cite{Ref:35}). 
\par Controllers developed for linear systems are first discussed. Control law, that predicts the input delay to the system is one the most studied as well as used method and has a root to the Smith predictor based approach \cite{Ref:6}. However, performance of this method depends on modelling accuracy of the system and can be applied to stable systems only. Artstein model reduction \cite{Ref:7}, finite spectrum assignment \cite{Ref:8}, on the contrary can be applied to unstable and multivariable plants.  The predictor based approach transforms the delayed system into a delay free system using the control input history. The prediction scheme of \cite{Ref:8} was extended in \cite{Ref:28} when state delay is also present along with the input delay. A discrete time predictor based controller is reported in \cite{Ref:9} for uncertain linear systems. A reduction method to design a delayed feedback based controller for uncertain linear systems with time varying input delay is proposed by \cite{Ref:10}. However, this approach requires exact knowledge of the input delay to compensate its effect. \cite{Ref:11} designed a robust controller for time varying input delay combining Lyapunov-Krasvoskii functional and a neutral transformation, assuming bound of the first derivative of delay is available. A predictor based state feedback controller is proposed in \cite{Ref:12}-\cite{Ref:13} for stabilization of uncertain discrete time systems via LMI technique. In the predictor based approach of \cite{Ref:19}, an equivalent transformation is used on the characteristic equation to convert the predictor feedback control into a stabilization problem of neutral time delay system for constant time delay. Predictor feedback control law, based on small gain analysis, for linear time invariant systems with unknown time varying delay is designed in \cite{Ref:27}.
 
\par Adaptive control techniques are addressed in \cite{Ref:14}-\cite{Ref:17} for linear systems. \cite{Ref:14} designed an adaptive control law for feedforward linear systems subjected to both state and input delay, where the unknown time delay is predicted adaptively using a projection function without considering plant uncertainties. \cite{Ref:16} considered linearly parametrized uncertainty in plants with known bounds. Adaptive backstepping method is used in \cite{Ref:17} for robust stabilization problem of linear non-minimum phase systems subjected to unknown dynamics and input delay. Sliding mode control strategies are considered in \cite{Ref:15}-\cite{Ref:30}. However, \cite{Ref:15} only considered constant time delay and uncertainties in plant parameters were ignored. \cite{Ref:30} used integral sliding mode control for uncertain linear systems under time varying delay at the input but it requires known bounds of the uncertainties as well as its first time derivative and can only tackle slowly varying input delay. On the other hand, \cite{Ref:30} used adaptive sliding mode control for linear input delayed systems where bound on the uncertainties are estimated online. However, estimation of the uncertainties becomes a monotonically increasing function of states for non zero values of states. This may cause very high switching gain and consequent chattering.

\par Nonlinear systems with input delay pose a greater challenge since linear boundedness of plant model cannot be incorporated in the stability proof \cite{Ref:5}. Compared to linear systems, results for nonlinear systems with time delay at the input is limited. \cite{Ref:20} provided a delay dependent Lyapunov-Razumikhin based approach on controlling nonlinear systems. Delay independent Razumikhin approach for nonlinear time delay systems with triangular structure is proposed in \cite{Ref:21}. The system was shown to be Uniformly Ultimately Bounded (UUB) assuming the delay disturbances are bounded. However, no control input was designed to negotiate the time delayed uncertainties and the delay independent solution is conservative. \cite{Ref:22} used Smith predictor based globally linearizing controller for known nonlinear system. \cite{Ref:24} proposed a backstepping based predictor approach for delay compensation of forward complete and strict-feedforward nonlinear systems utilizing ODE-PDE cascade transformation.  \cite{Ref:25}-\cite{Ref:26} developed controllers for rigid and flexible link manipulators by linearizing the system. Here, the controllers developed for delay free system is shown to be stable under some delay dependent conditions.  However, the above mentioned control laws require exact knowledge of the models which is difficult for nonlinear systems. This issue was addressed in \cite{Ref:31} where a predictor based approach is proposed for constant input delay based on Lyapunov-Krasvoskii method. In this case, uncertainty in the inertia matrix was considered. The same approach was later extended for time varying input delay in \cite{Ref:32} and also for state delay \cite{Ref:33}. A filtered tracking error based control law is used in \cite{Ref:34} to tackle the effect of known constant input time delay and used fuzzy logic systems of first type to approximate the unknown disturbances with predefined uncertainty bound. A stabilizing controller for spacecraft with unknown but constant input delay is proposed in \cite{Ref:35} with approximated uncertainty bound. However, in practice delay may not always be constant and defining prior uncertainty bounds for the uncertainties is not always feasible for nonlinear uncertain systems due to unmodelled dynamics and unknown disturbances.


\subsection{Contribution}
In this paper, a delay dependent control law, referred as Adaptive Robust Outer Loop Control (AROLC), is proposed for tracking control of a class of Euler-Lagrange systems subjected to unknown time-varying input delay and bounded uncertainty. The two major advantages of the proposed control scheme is stated below, \\
\par (i) The proposed control law does not require any explicit knowledge of the delay or its time derivatives to design the control law unlike the ones reported in \cite{Ref:31}-\cite{Ref:34}. It can negotiate any input delay within a maximum delay range, obtained from the delay dependent stability criterion of the closed loop system employing the Razumikhin-type stability approach \cite{Ref:37}.\\
\par (ii) As opposed to the controllers developed in \cite{Ref:31}-\cite{Ref:35}, the proposed adaptive-robust law does not necessitate the knowledge of any predefined uncertainty bound or its modelling. The switching gain adapts itself with an error function defined by the user. The proposed controller possesses the flexibility that a user can specify any error function based on which the switching gain would be modified while maintaining the same stability criterion. The closed loop stability analysis of the system dynamics is carried out in the sense of UUB.
\par (iii) Experimental validation of the proposed control law is also carried out on PIONEER-3 wheeled mobile robot (WMR) and compared the results with the predictive control approaches reported in \cite{Ref:31}-\cite{Ref:32}.

\subsection{Notations}
The following notations are used throughout the paper: any quantity $\rho$ delayed by an amount $h$ as $\rho(t-h)$, would be denoted as $\rho_h$; $\lambda _{min}(\cdot)$ and $|| \cdot ||$ represent minimum eigen value and Euclidean norm of the argument respectively; $I$ represents identity matrix.

\subsection{Organization}
The article is organized as follows: the detailed problem formulation is first carried out in Section 2. This is followed by the proposed adaptive-robust control methodology and its detail analysis. Section 3 presents the experimental results of the proposed controller and its comparison with \cite{Ref:32}. Section 4 concludes the entire work.

\section{Controller Design}
\subsection{Problem Formulation and Objective Definition}
In general, an Euler-Lagrange system with second order dynamics, with input delay, can be written as,
\begin{equation}\label{sys new}
M(q)\ddot{q}+N(q,\dot{q})=\tau(t-h(t)),
\end{equation}
where, $q(t)\in\Re^{n}$ is the system state, $\tau(t)\in\Re^{n}$ is the control input, $M(q)\in\Re^{n\times n}$  is the mass/inertia matrix and $N(q,\dot{q})\in\Re^{n}$ denotes combination of other system dynamics terms based on system properties. In practice, it can be assumed that unmodelled dynamics and disturbances is subsumed by $N$. $h(t)$ is a time varying delay at the input $\tau$. The delay may be result of computation delay, communication delay between the controller and the actuator etc. Let, $q^d(t)$ be the desired trajectory to be tracked and $e_1(t)=q^d(t)-q(t)$ is the tracking error. Before introducing the proposed control law, brief overview of the control structure of predictive control law, reported in \cite{Ref:32} is provided below:

\subsubsection{Predictive Controller (\cite{Ref:31}-\cite{Ref:32})}

\begin{assum}\label{assum 6}
All the disturbances and their first time derivative is bounded by some known constant. 
\end{assum}

\begin{assum}\label{assum 7}
Input delay $h(t)$ and $\ddot{h}(t)$ is bounded by a known constant and $\dot{h}(t)<1$. 
\end{assum}

Let $\varrho$ be a measurable filtered tracking error and defined as
\begin{align}
\varrho&=\dot{e}_1+\kappa e_1-\vartheta e_z \label{filter}\\ 
e_z& = \int_{t-h(t)}^{t} \tau (\theta) d\theta \label{int error}
\end{align}
where, $\kappa$ is a positive constant and $\vartheta$ is a positive definite matrix. Taking time derivative of (\ref{filter}) and multiplying it by $M$ and then using (\ref{sys new}) and (\ref{int error}) we get,
\begin{equation}
M\dot{\varrho}=M\ddot{q}^d+N-M \eta (\tau-\tau_h+ \dot{h}\tau_h)-\tau-\dot{h}\tau_h+ \kappa M \dot{e}_1,
\end{equation}
where, $\eta=\vartheta-M^{-1}$. The control input is selected as $\tau=k_b \varrho$ to track the desired trajectory provided Assumptions 1 and 2 holds. Selection methods of $k_b$ and $\kappa$ are detailed in \cite{Ref:31}-\cite{Ref:32}.

\begin{remark}
 In practice, proper bound estimation of unmodelled dynamics and unknown disturbances is almost impossible. Moreover, according to Assumption 2, knowledge of the bound of $\dot{h}$ and $\ddot{h}$ is required and the controller can only tackle slowly varying time delay \cite{Ref:32}. However, this is difficult to attain since variation in time delay $h$ can be arbitrary in practical scenario. Hence, it is not always possible to maintain Assumptions 1 and 2.
\end{remark} 
  So, the aim of this paper is to design a control law which simultaneously fulfils the following two objectives:
\par \textbf{O1}:  To maintain system stability of the input delayed system (\ref{sys new}), within a stipulated maximum time delay, while it follows a predefined desired path. This objective is to be met irrespective and without any knowledge of the nature of the variation delay $h(t)$.
\par \textbf{O2}:  To provide robustness to system (\ref{sys new}) against bounded but unknown uncertainties and disturbances.  without any knowledge of the uncertainties and its bounds. 
 

\subsection{Adaptive Robust Outer Loop Controller}
Towards achieving the outlined objectives, a novel controller, named Adaptive Robust Outer Loop Controller (AROLC) is proposed. The structure of the the proposed control law is selected to be,
\begin{equation}\label{input}
\tau=\hat{M}u+\hat{N},
\end{equation}
where, $u$ is the auxiliary control input; $\hat{M}$ and $\hat{N}$ are the nominal values of $M$ and $N$ respectively. $u$ is defined in the following way,


\begin{equation}\label{control input}
u=\hat{u}+\Delta u,
\end{equation}
where, $\hat{u}$ and $\Delta u$ are nominal and switching control input respectively and they are evaluated as, 
\begin{equation}\label{aux}
\hat{u}=\ddot{q}^d(t)+K_2\dot{e_1}(t)+K_1e_1(t),
\end{equation}
\begin{equation}\label{delta u}
 \Delta u=
  \begin{cases}
    \alpha\hat{c}(e,t)\frac{s}{\parallel s \parallel}       & \quad \text{if } \parallel s \parallel\geq \epsilon,\\
    \alpha\hat{c}(e,t)\frac{s}{\epsilon}        & \quad \text{if } \parallel s \parallel< \epsilon,\\
  \end{cases}
\end{equation}
where, $K_1$ and $K_2$ are two positive definite matrices with appropriate dimensions, $s=B^TPe$, $e=\begin{bmatrix}
e_1^T & \dot{e}_1^T
\end{bmatrix}^T$, $\hat{c}$ is the switching gain responsible to tackle the uncertainties, $ \alpha>0$ is a scalar adaptive gain and $\epsilon>0$ represents a small scalar. To accomplish the second objective, the following novel adaptive control law for evaluation of $\hat{c}$ is proposed which eliminates the requirement of any predefined bound of uncertainties:


\begin{equation}\label{ATRC}
 \dot{\hat{c}} =
  \begin{cases}
   \quad  \parallel s \parallel     & \quad \hat{c}>\gamma, s^T\dot{s} >0\\
    -\parallel s \parallel       & \quad \hat{c}>\gamma,s^T\dot{s} \leq 0 \\
    \qquad  \gamma        & \quad \hat{c}\leq\gamma,\\
  \end{cases}
\end{equation}

where, $\gamma>0$ is a small scalar to keep $\hat{c}$ always positive. According to the adaptive law (\ref{ATRC}), $\hat{c}$ increases (resp. decreases) whenever the error trajectories move away from (resp. move closer to) $ s =0 $. Nevertheless, due to the delay $h$ the error dynamics of (\ref{sys new}) employing (\ref{input}) and (\ref{control input}) is found to be,


\begin{equation}\label{error dyn delayed 1}
\ddot{e}_1=-K_2\dot{e}_{1h}-K_1e_{1h}+\sigma-\Delta u_h,
\end{equation}
where, $\sigma=(I-M^{-1}(q)\hat{M}(q_h))u_h+M^{-1}(q)(N(q,\dot{q})-\hat{N}(q_h,\dot{q}_h))+\ddot{q}^d(t)-\ddot{q}^d_h$. Further, (\ref{error dyn delayed 1}) can be formulated in state space as,
\begin{equation}\label{new err dyn 1}
\dot{e}=A_1e+B_1e_h+B(-\Delta u_h+\sigma),
\end{equation}
where, $ A_1=\begin{bmatrix}
0 & I \\
0 & 0
\end{bmatrix} 
, 
B_1=\begin{bmatrix}
0 & 0\\
-K_1 & -K_2
\end{bmatrix}$. Noting that, $e_h=e(t)-\int\limits_{-h}^0 \dot{e}(t+\theta)d\theta$, where the derivative inside the integral is with respect to $\theta$, the error dynamics (\ref{new err dyn 1}) is modified as,
\begin{equation}\label{error dyn delayed}
\dot{e}(t)=Ae(t)-B_1\int\limits_{-h}^0 \dot{e}(t+\theta)d\theta+B(-\Delta u_h+\sigma),
\end{equation}
where, $A=A_1+B_1$. The controller gains $K_1$ and $K_2$ are selected in a way such that $A$ is Hurwitz which is always possible by noting the structure of $A$. The stability analysis of the proposed controller is carried out in the sense of UUB and is detailed subsequently.

\subsubsection{Stability Analysis of AROLC}
\begin{assum}\label{assum 1}
All the states i.e $q, \dot{q}$ are available.
\end{assum}
\begin{assum}\label{assum 2}
The uncertainties are bounded as $|| \sigma || \leq c$, where $c$ is an unknown scalar quantity. Knowledge of $c$ is not required to compute the proposed control law. However, consideration of $c$ is necessary for stability analysis. 
\end{assum}
\begin{assum}\label{assum 5}
Let, V(e) be a Lyapunov function candidate. Then  following the Razumikhin-type theorem \cite{Ref:37}, for some constant $r>1$, let the following inequality holds:
\begin{equation}\label{razu}
V(e(\xi ))<rV(e(t)),\qquad t-2h\leq \xi \leq t.
\end{equation}
\end{assum}
Consider a Lyapunov function of the following form:
\begin{equation}\label{lyapunov}
V(e)=V_1(e)+ V_2(e),
\end{equation}
where, $V_1(e)=\frac{1}{2}e^TPe$, $V_2(e)=\frac{1}{2}(\hat{c}-c)^2$,  $P>0$ is the solution of the Lyapunov equation $A^TP+PA=-Q$ for some $Q>0$.

\begin{theorem}
The maximum allowable delay time for system (\ref{sys new}) employing the control input (\ref{input})-(\ref{delta u}) is found to be,
\begin{equation}\label{delay value}
h< \frac{\lambda _{min}(Q)}{||\beta PB_1( A_1P^{-1}A_1^T+B_1P^{-1}B_1^{T}+P^{-1} )B_1^{T}P+2(r/\beta )P)||},
\end{equation}
provided gains $K_1$, $K_2$ and scalar design parameter $r>1, \beta>0$ are selected in a manner which satisfies 
\begin{equation}
\lambda _{min}(Q)>h||\beta PB_1( A_1P^{-1}A_1^T+B_1P^{-1}B_1^{T}+P^{-1} )B_1^{T}P+2(r/\beta )P)|| \quad \forall h.
\end{equation}
\begin{proof}:
Using (\ref{error dyn delayed}), the time derivative  of $V_1(e)$ yields,
\begin{equation}\label{lya_dot for time delay}
\dot{V}_1=-\frac{1}{2}e^TQe-e^{T}PB_1\int_{-h}^{0}\dot{e}(t+\theta)d\theta+s^T(-\Delta u_h+\sigma)
\end{equation}
Again using (\ref{new err dyn 1}),
\begin{align}
-e^{T}PB_1\int_{-h}^{0}\dot{e}(t+\theta)d\theta=- \int_{-h}^{0}e^{T}PB_1 & [ A_1e(t+\theta ) +B_1e(t-h+\theta )\nonumber\\
& +B\sigma _2(t+\theta ) ]d\theta, \label{relation}
\end{align}
where, $\sigma_2(t)=-\Delta u_h+\sigma(t)$. Applying (\ref{razu}) to (\ref{lyapunov}), the following relation is obtained,
\begin{equation}
e^{T}(\xi )Pe(\xi )<re^{T}(t)Pe(t)+\varphi(\xi),
\end{equation}  
where, $\varphi(\xi)=r(\hat{c}(t)-c)^{2}-(\hat{c}(\xi )-c)^{2}$. For any two non zero vectors $z_1$ and $z_2$, there exists a constant $\beta>0$ and matrix $D>0$ such that the following inequality holds,
\begin{equation}\label{ineq 1}
-2z_1^{T} z_2\leq \beta z_1^{T}D^{-1}z_1+(1/\beta ) z_2^{T}D z_2.
\end{equation}
Applying (\ref{ineq 1}) to (\ref{relation}) and taking $D=P$ the following inequalities are obtained,
\begin{align}\label{cond1}
- 2\int_{-h}^{0}e^{T}PB_1 A_1\left [ e(t+\theta)\right ]d\theta &\leq \int_{-h}^{0}[ \beta e^{T}PB_1A_1P^{-1}A_1^{T}B_1^{T}Pe\nonumber\\
&\qquad \qquad \qquad +(1/\beta )e^{T}(t+\theta )Pe(t+\theta ) ]d\theta \nonumber\\
&\leq h e^{T}\left [ \beta PB_1A_1P^{-1}A_1^{T}B_1^{T}P+(r/\beta )P) \right ]e\nonumber\\
&\qquad \qquad \qquad +\int_{-h}^{0}(1/\beta)\varphi(t+\theta) d\theta 
\end{align}
\begin{align}\label{cond2}
- 2\int_{-h}^{0}e^{T}PB_1B_1 \left [  e(t-h+\theta)\right ]d\theta &\leq \int_{-h}^{0} [ \beta e^{T}PB_1B_1P^{-1}B_1^{T}B_1^{T}Pe\nonumber\\
& \qquad \quad+\frac{1}{\beta} e^{T}(t-h+\theta )Pe(t-h+\theta ) ]d\theta \nonumber \\
& \leq h e^{T}\left [ \beta PB_1B_1P^{-1}B_1^{T}B_1^{T}P+\frac{r}{\beta}P) \right ]e\nonumber \\
& \qquad \qquad +\int_{-h}^{0}(1/\beta)\varphi(t-h+\theta ) d\theta 
\end{align}
\begin{align}\label{cond3}
- 2\int_{-h}^{0}e^{T}PB_1 \left [ B \sigma _2(t+\theta)\right ]d\theta &\leq \int_{-h}^{0}[ \beta e^{T}PB_1P^{-1}B_1^{T}Pe\nonumber\\
& \qquad \qquad +\frac{1}{\beta}(B \sigma _2(t+\theta))^{T}PB \sigma _2(t+\theta)   ]d\theta \nonumber \\
& \leq h e^{T}\left [ \beta PB_1P^{-1}B_1^{T}P \right ]e\nonumber \\
& +\int_{-h}^{0}\frac{1}{\beta} (B \sigma _2(t+\theta))^{T}PB \sigma _2(t+\theta)  d\theta 
\end{align}
Since $P>0$, we can write $P=\bar{P}^T\bar{P}$ for some $\bar{P}>0$. Then, assuming the uncertainties to be square integrable within the delay, there exists a scalar $\Gamma>0$ such that, the following inequality holds over the stipulated time delay:
\begin{align}\label{cond4}
\frac{1}{2\beta }\left \| \int_{-h}^{0}\left [ \varphi(t+\theta)+\varphi(t-h+\theta )+(B \sigma _2(t+\theta))^{T}\bar{P}^T\bar{P}B \sigma _2(t+\theta) \right ]d\theta  \right \|\leq \Gamma. 
\end{align} 
Substituting (\ref{cond1})-(\ref{cond4}) into (\ref{lya_dot for time delay}) yields,
\begin{equation}\label{delay cal}
\dot{V}_1(e) \leq -\frac{1}{2}e^{T}\left [ Q-hE  \right ]e+ \Gamma+ s^{T}(-\Delta u+\sigma)+s^{T}\Theta, 
\end{equation}
where, $E =\left ( \beta PB_1\left ( A_1P^{-1}A_1^T+B_1P^{-1}B_1^{T}+P^{-1} \right )B_1^{T}P+2(r/\beta )P \right )$, $\Theta=\Delta u(t)-\Delta u_h$.
To achieve negativeness of the first term of (\ref{delay cal}), the controller gains $K_1,K_2$ and the scalar variables $r, \beta$ are to be selected in a manner such that $\lambda _{min}(Q)>h|| E || \quad \forall h$. Hence, the maximum allowable delay can be found from (\ref{delay cal}) as,
\begin{equation}\label{delay value 1}
h< \frac{\lambda _{min}(Q)}{|| E ||}.
\end{equation}
Since $E$ is governed by the controller gains, (\ref{delay value}) provides the maximum input delay that system (\ref{sys new}) can sustain for some choice of $K_1,K_2,Q,r$ and $\beta$.
\end{proof}
\end{theorem}

Exploring the various possible combinations of $\Delta u$ and $\hat{c}$ in (\ref{delta u}) and (\ref{ATRC}) respectively, six different cases have been identified. Stability analysis of (\ref{sys new}) employing AROLC is stated in Theorem 1. In each case, UUB is established using a common Lyapunov function.

\begin{theorem}
The system (\ref{sys new}) employing the control input (\ref{input}), (\ref{control input}), having nominal control input (\ref{aux}) and switching control input (\ref{delta u}) with adaptation law (\ref{ATRC}) is UUB.
\end{theorem}

\begin{proof}: Taking $\Psi=Q-hE>0$ and utilizing result of Theorem 1, (\ref{control input})-(\ref{ATRC}) and (\ref{lyapunov}), each of the six cases is analysed as follows:
\par \textbf{Case (i):} \quad $\hat{c}>\gamma, \quad || s || \geq \epsilon \quad \text{and} \quad s^T\dot{s}>0$\\
 Utilizing (\ref{ATRC}) and (\ref{delay cal}), the time derivative of (\ref{lyapunov}) yields,
\begin{align}
\dot{V}(e)&\leq-\frac{1}{2}e^{T}\Psi e+ \Gamma+ s^{T}(-\Delta u+\sigma)+s^{T}\Theta+(\hat{c}-c)||s||\nonumber\\
& = -\frac{1}{2}e^{T}\Psi e+ \Gamma+ s^{T}(-\alpha\hat{c}\frac{s}{||s||}+\sigma)+s^{T}\Theta+(\hat{c}-c)||s||.\label{proof eq}
\end{align}
Again, 
\begin{equation}\label{case 1}
s^T(-\alpha\hat{c}\frac{s}{|| s ||}+\sigma)=-\alpha\hat{c}\frac{s^Ts}{|| s ||}+s^T\sigma \leq (-\alpha\hat{c}+c)|| s ||
\end{equation}
Using (\ref{proof eq}) and (\ref{case 1}) we have,
\begin{align*}
\dot{V}(e) &\leq  -\frac{1}{2}\lambda_{min}(\Psi)|| e|| ^2-(\alpha-1)\hat{c}|| s ||+||\Theta|||| s ||+\Gamma
\end{align*}
So, for a choice of $\alpha>1$, $\dot{V}(e)<0$ would be established if $\lambda _{min}(\Psi)|| e ||^{2} > 2 \Gamma +2|| \Theta || || s ||$. Again, we have $|| s || \leq ||B^TP|| ||e||$. Thus (\ref{sys new}) would be UUB with the ultimate bound,
\begin{equation}\label{error bound 1}
||e||=\mu_1+\sqrt{\frac{2\Gamma}{\lambda _{min}(\Psi)}+\mu_1^2}=\varpi_1.
\end{equation}
where, $\mu_1=\frac{|| \Theta ||||B^TP||}{\lambda _{min}(\Psi)}$.
Let $\Xi$ denote the smallest level surface of $V$ containing the ball $B_{\varpi_1}$ with radius $\varpi_1$ centred at $e=0$. For initial time $t_0$, if $e(t_0)\in \Xi$ then the solution remains in $\Xi$. If  $e(t_0)\notin \Xi$ then $V$ decreases as long as $e(t)\notin \Xi$. The time required to reach $\varpi_1$ is zero when $e(t_0)\in \Xi$, otherwise, while $e(t_0)\notin \Xi$ the finite time $t_{r1}$ to reach $\varpi_1$ is given by \cite{Ref:36},
\begin{equation}\label{time}
t_{r1}-t_0\leq (\parallel e(t_0) \parallel-\varpi_1)/c_0, 
\end{equation}
where, $\dot{V}(t)\leq-c_0$ for some $c_0>0$.
\par \textbf{Case (ii):} \quad $\hat{c}>\gamma, \quad || s || \geq \epsilon \quad \text{and} \quad s^T\dot{s}<0$ \\
Exploring (\ref{delay cal}), (\ref{ATRC}) and (\ref{case 1})  the time derivative of (\ref{lyapunov}) yields,

\begin{align*}
\dot{V}(e) &\leq -\frac{1}{2}e^T \Psi e+(c-\alpha\hat{c})|| s ||+s^{T}\Theta+\Gamma-(\hat{c}-c)||s||\\
&\leq -\frac{1}{2}\lambda_{min}(\Psi)|| e|| ^2+(2c-(\alpha+1)\hat{c}+||\Theta||)|| s ||+\Gamma
\end{align*}

So, $\dot{V}(e)<0$ would be achieved if $\lambda _{min}(\Psi)|| e ||^{2} > 2(2c-(\alpha+1)\hat{c}+||\Theta||)|| s ||+2\Gamma$. Thus (\ref{sys new}) would be UUB with the ultimate bound,
\begin{equation}\label{error bound 1}
||e||=\mu_2+\sqrt{\frac{2\Gamma}{\lambda _{min}(\Psi)}+\mu_2^2}=\varpi_2.
\end{equation}
where, $\mu_2=\frac{(2c-(\alpha+1)\hat{c}+||\Theta||)||B^TP||}{\lambda _{min}(\Psi)}$.

\par \textbf{Case (iii):} \quad $\hat{c} \leq \gamma, \quad || s || \geq \epsilon  \quad \text{and any} \quad s^T\dot{s} $\\
Since $\hat{c}\leq \gamma$ we have $(\hat{c}-c )\gamma \leq \gamma^2-c\gamma \leq \gamma^2$. So, using (\ref{ATRC}) and following similar procedure we have,
\begin{align*}
\dot{V}(e)& \leq -\frac{1}{2}e^{T}\Psi e+ \Gamma+ s^{T}(-\Delta u+\sigma)+s^{T}\Theta+(\hat{c}-c )\gamma\nonumber\\
& \leq -\frac{1}{2}e^{T}\Psi e+ \Gamma+ s^{T}(-\alpha\hat{c}\frac{s}{||s||}+\sigma)+|| s ||||\Theta||+\gamma^2.
\end{align*}
Using (\ref{case 1}) we have,
\begin{equation*}
\dot{V}(e) \leq -\frac{1}{2}\lambda_{min}(\Psi)|| e|| ^2+(c-\alpha\hat{c}+||\Theta||)|| s ||+\Gamma+\gamma^2
\end{equation*}
So, system would be UUB with the following ultimate bound,
\begin{equation}
||e||=\mu_3+\sqrt{\frac{2(\Gamma+\gamma^2)}{\lambda _{min}(\Psi)}+\mu_3^2}=\varpi_3.
\end{equation}
where, $\mu_3=\frac{( c-\alpha\hat{c}+|| \Theta  || )}{\lambda _{min}(\Psi)}$.

\par \textbf{Case (iv):} \quad $\hat{c} > \gamma, \quad || s || < \epsilon \quad \text{and} \quad s^T\dot{s}>0$\\
The term $s^T\sigma$ turns out to be:
\begin{equation}\label{for case iv}
s^T\sigma \leq \Vert s \Vert \Vert\sigma \Vert \leq c\Vert s \Vert=c\frac{s^Ts}{\Vert s \Vert}
\end{equation}
Using  (\ref{ATRC}), (\ref{delay cal}) and (\ref{for case iv}), the time derivative of (\ref{lyapunov}) gives,

\begin{align}
\dot{V}(e)&\leq-\frac{1}{2}e^{T}\Psi e+ \Gamma+ s^{T}(-\Delta u+\sigma)+s^{T}\Theta+(\hat{c}-c)||s||\nonumber\\
& \leq -\frac{1}{2}\lambda_{min}(\Psi)|| e|| ^2+ \Gamma+ s^{T}( -\alpha\hat{c}\frac{s}{\epsilon} +c\frac{s}{|| s ||} )+||s||||\Theta||+(\hat{c}-c)||s||.\label{case iv cond 1}
\end{align}
The combination of third, fourth and fifth term of (\ref{case iv cond 1}) takes the maximum value of $(\epsilon(\hat{c}+||\Theta||)^2 )/4\alpha\hat{c}$ for $|| s ||=(\epsilon(\hat{c}+||\Theta||))/2\alpha\hat{c}$. Thus $\dot{V}(e)<0$ would be achieved if $\lambda _{min}(\Psi)|| e ||^{2} > 2\Gamma+(\epsilon(\hat{c}+||\Theta||)^2 )/2\alpha\hat{c}$. So, the system is UUB with the ultimate bound is calculated to be,
\begin{equation}
|| e ||=\sqrt{\frac{4\alpha\Gamma\hat{c}+\epsilon (\hat{c}+||\Theta|| )^2}{2\alpha \hat{c} \lambda_{min}(\Psi)}}=\varpi_4
\end{equation}

\par \textbf{Case (v):} \quad $\hat{c} > \gamma, \quad || s || < \epsilon \quad \text{and} \quad s^T\dot{s}<0$\\
Using  (\ref{ATRC}), (\ref{delay cal}) and (\ref{for case iv}),
\begin{align}
\dot{V}(e)& \leq-\frac{1}{2}e^{T}\Psi e+ \Gamma+ s^{T}(-\Delta u+\sigma)+s^{T}\Theta-( \hat{c}-c)|| s ||\nonumber\\
& \leq -\frac{1}{2}\lambda_{min}(\Psi)|| e|| ^2+ \Gamma+ s^{T}( -\alpha\hat{c}\frac{s}{\epsilon} +c\frac{s}{|| s ||} )+||s|| ||\Theta||-( \hat{c}-c)|| s ||\label{case v}
\end{align}

The combination of third, fourth and fifth term of (\ref{case v}) takes the maximum value of $(\epsilon(2c-\hat{c}+||\Theta||)^2 )/4\alpha\hat{c}$ for $|| s ||=\left(\epsilon(2c-\hat{c}+||\Theta||) \right)/2\alpha\hat{c}$. Thus $\dot{V}(e)<0$ would be achieved if $\lambda _{min}(\Psi)|| e ||^{2} > 2\Gamma+(\epsilon(2c-\hat{c}+||\Theta||)^2 )/2\alpha\hat{c}$. So, the system is UUB and the ultimate bound is calculated to be,
\begin{equation}
|| e ||=\sqrt{\frac{4\alpha\Gamma\hat{c}+\epsilon (2c-\hat{c}+||\Theta||)^2}{2\alpha \hat{c} \lambda_{min}(\Psi)}}=\varpi_5
\end{equation}


%
%
%

\par\textbf{Case (vi):} \quad $\hat{c} \leq \gamma, \quad || s || < \epsilon  \quad \text{and any} \quad s^T\dot{s} $

Following similar procedure, time derivative of (\ref{lyapunov}) for Case (vi) yields,
\begin{align}
\dot{V}(e)&\leq-\frac{1}{2}e^{T}\Psi e+ \Gamma+ s^{T}(-\Delta u+\sigma)+s^{T}\Theta+\gamma^2\nonumber\\
& \leq -\frac{1}{2}\lambda_{min}(\Psi)|| e|| ^2+ \Gamma+ s^{T}( -\alpha\hat{c}\frac{s}{\epsilon} +c\frac{s}{|| s ||} )+||s||||\Theta||+\gamma^2.\label{case vi cond 1}
\end{align}
The combination of third and fourth term of (\ref{case vi cond 1}) takes the maximum value of $(\epsilon(c+||\Theta||)^2 )/4\alpha\hat{c}$ for $\Vert s \Vert=\left(\epsilon\left(c+||\Theta|| \right) \right)/2\alpha\hat{c}$. Thus $\dot{V}(e)<0$ would be achieved if $\lambda _{min}(\Psi)|| e ||^{2} > 2\Gamma+(\epsilon(c+||\Theta||)^2 )/2\alpha\hat{c}+2\gamma^2$. So, the system is UUB with the ultimate bound is defined as,
\begin{equation}
|| e ||=\sqrt{\frac{4\alpha\hat{c}(\Gamma+\gamma^2)+\epsilon (c+||\Theta|| )^2}{2\alpha \hat{c} \lambda_{min}(\Psi)}}=\varpi_6
\end{equation}
The finite reaching time to the error bounds $\varpi_i, i=2,\cdots,6$ can be computed similarly from (\ref{time}).

\end{proof}

\begin{remark}
The performance of AROLC can be characterized by the various error bounds ($\varpi_i, i=1,\cdots,6$) under various conditions. The scalars $\varpi_i's$ are function of $\alpha$ and delay $h$. It can be seen that as $\alpha$ increases and $h$ drops, better tracking accuracy can be achieved. However, too large a $\alpha$ may result in high control input. Also, one may choose different values of $\alpha$ for $s^T\dot{s} \leq 0$ and $s^T\dot{s}>0$. Again, it is to be noticed that, instead of $s^T\dot{s}$ user can select any suitable error function while keeping the same stability notion. 
\end{remark}
\textit{\textbf{Comparison with the existing results:}}
\begin{itemize}

\item  It can be observed from (\ref{delay value}) that, proposed stability approach is independent of the rate of change in $h$ and thus can negotiate any arbitrarily time-varying delay within the stipulated maximum bound. On the contrary, \cite{Ref:32}-\cite{Ref:33} can negotiate only slowly varying time delay and also bound of $\ddot{h}$ is required. 

\item Computation of switching gain $\hat{c}$ of AROLC using (\ref{ATRC}), unlike \cite{Ref:31}-\cite{Ref:35}, does not require predefined bound of the uncertainties and helps to to attain the tracking objective. Accomplishment of this objective in turn reduces the tedious modelling effort of complex nonlinear systems. To illustrate the fact with an example, $\hat{N}$ does not need to include friction, slip, skid etc. for wheeled mobile robot and these terms can be treated as uncertainties.

\end{itemize}

\section{Application: Nonholonomic WMR}
Nonholonomic WMR provides a unique platform to test the proposed control law since under practical circumstances, a WMR is always subjected to uncertainties like friction, slip, skid etc. These terms are extremely difficult to model and in many cases they are not considered while modelling. The dynamic equation of a WMR after solving the Lagrange multiplier can be written as follows \cite{Ref:37},
\begin{equation}
\bar{M}(q)\ddot{q}+\bar{V}(q,\dot{q})=Gu,
\end{equation}
where,
\begin{align*}
\bar{M}=&\begin{bmatrix}
m & 0 & Ksin\phi & k_1 & k_2 \\ 
0 & m & -Kcos\phi & k_3 & k_4\\ 
Ksin\phi & -Kcos\phi & \bar{I} & -k_5 & k_5 \\ 
k_1 & k_3 & -k_5 & I_w & 0 \\ 
k_2 & k_4 & k_5 & 0 & I_w
\end{bmatrix},u=\begin{bmatrix}
u_r\\ 
u_l
\end{bmatrix}\\
G=&\begin{bmatrix}
0 &0 \\ 
0 &0 \\ 
0 & 0\\ 
1 & 0\\ 
0 & 1
\end{bmatrix},
\bar{V}=\begin{bmatrix}
md\dot{\phi}^2cos\phi+m\bar{r}^2sin\phi(\dot{\theta }_{\bar{r}}^{2}-\dot{\theta }_{l}^{2})/2b\\ 
md\dot{\phi}^2cos\phi-m\bar{r}^2sin\phi(\dot{\theta }_{\bar{r}}^{2}-\dot{\theta }_{l}^{2})/2b\\ 
K\bar{r}^2(\dot{\theta }_{\bar{r}}^{2}-\dot{\theta }_{l}^{2})/2b\\ 
-K\bar{r}\dot{\phi}^2/2\\ 
-K\bar{r}\dot{\phi}^2/2\\ 
\end{bmatrix}\\
k_1=&sin\phi (md\bar{r}-K\bar{r})/b-m\bar{r}cos\phi /2,\\
k_2=&sin\phi (K\bar{r}-md\bar{r})/b-m\bar{r}cos\phi /2,\\
k_3=&cos\phi (K\bar{r}-md\bar{r})/b-m\bar{r}sin\phi /2,\\
k_4=&cos\phi (md\bar{r}-K\bar{r})/b-m\bar{r}sin\phi /2,
k_5=\bar{r}(I-Kd)/b,
\end{align*}
Here $q\in\Re^5=\lbrace x_c,y_c, \phi, \theta_r, \theta_l \rbrace$ is the generalized coordinate vector of the system. The position of the WMR can be specified by three generalized coordinates $(x_c,y_c, \phi)$ where, $(x_c,y_c)$ are the coordinates of the center of mass of the system and $\phi$ is the heading angle. $(\theta_r, \theta_l)$ and $(u_r,u_l)$  are rotation and torque inputs of the right and left wheels respectively. $m$ and $\bar{I}$ represents the mass and inertia of the overall system. Definition of other system parameters are detailed in \cite{Ref:38}. 

\subsection{Experimental Results and Comparison}

\begin{figure*}[t]
\begin{center}
\includegraphics[width=4.5 in,height=2.1 in]{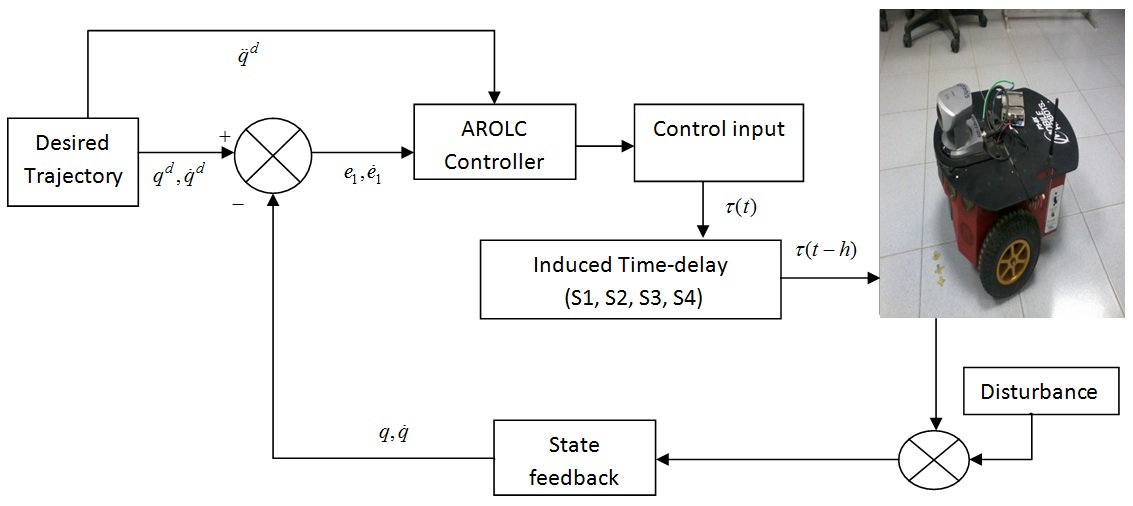}
\caption{Control architecture for TARC while employing on PIONEER-3.}\label{fig:block}
\end{center}
\end{figure*}

AROLC is employed in "PIONEER-3" WMR while the robot is directed to track the following circular path:
\begin{align*}
x_c^d&=1.25sin(0.35t)+.1,\\
y_c^d&=1.25cos(0.35t)+1.35,\\
\phi^d&=0.25t, \theta_r^d=3t, \theta_l^d=2t.
\end{align*}

 The control architecture of the proposed AROLC is depicted in Fig. \ref{fig:block}. Tracking performance of AROLC is compared with the Predictive Controller (PCON)reported in \cite{Ref:32}. For a choice of $K_1=K_2=Q=I, \beta=1, r=1.1$, the maximum allowable delay is found to be $h_m=125 ms$. Other parameters to design AROLC are defined as $\alpha=2, \epsilon=0.1, \gamma=0.001$. The following time-varying input delay was induced into the system for both the controllers: 
\par \textbf{S1}: $h(t)=20+80$ $abs(sin(t))$ ms
\par \textbf{S2}: $h(t)=5+120$ $abs(sin(0.1t))$ ms.
 

The input delay is envisaged by halting the programme for $h(t)$ amount of time in the VC++ programming environment between the control input computation and feeding it to the system. Again, to create a dynamic payload variation, a further $3.5 kg$ payload is added and kept for $5 sec$ and then removed. This process is carried out for the entire duration of experimentation. A time gap $5 sec$ is maintained between two successive instances of addition of the payload. However, the payload was added randomly at different places on the robotic platform every time to create dynamic variation in center of mass and inertia. 
\par The trajectory tracking performance between AROLC and PCON is depicted in Fig. \ref{fig:1} for the condition \textbf{S1}. The corresponding $x_c$ and $y_c$ position error comparison is illustrated through Fig. \ref{fig:2} and Fig. \ref{fig:3} respectively. All the error plots in this paper are in absolute value. In the situation \textbf{S2}, Fig. \ref{fig:5} and Fig. \ref{fig:6} depict the $x_c$ and $y_c$ position error comparison respectively. Due to the robustness property against the unmodelled dynamics, AROLC provides better accuracy over PCON through its switching control logic. This can easily be comprehended from these error plots. To infer the performance of the individual controllers, absolute average error in $x_c$ (AE-$x_c$) and $y_c$ position (AE-$y_c$) is provided in Table \ref{table 1} for \textbf{S1} and \textbf{S2}. The percentage error is calculated with respect to the diameter of the circular path. The tabulated data further establishes the superior performance of AROLC over PCON.


\begin{figure}[t]
\begin{center}
\includegraphics[width=3 in,height=2.7 in]{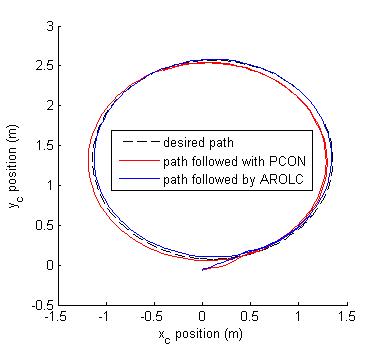}
\caption{Trajectory tracking performance comparison for input delay \textbf{S1}.}\label{fig:1}
\end{center}
\end{figure}

\begin{figure}[t]
\begin{center}
\includegraphics[width=3 in,height=2.1 in]{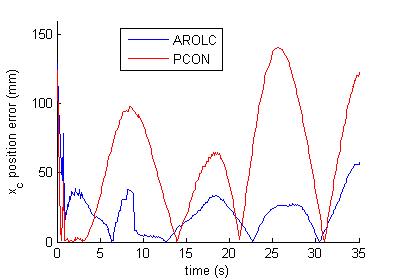}
\caption{ $x_c$ position error comparison for input delay \textbf{S1}.}\label{fig:2}
\end{center}
\end{figure}

\begin{figure}[t]
\begin{center}
\includegraphics[width=3 in,height=2.1 in]{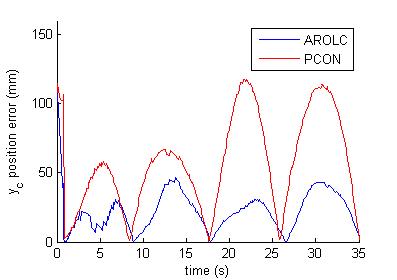}
\caption{$y_c$ position error comparison for input delay \textbf{S1}.}\label{fig:3}
\end{center}
\end{figure}

\begin{table}[!t]
\renewcommand{\arraystretch}{1.3}
\caption{$x_c$ and $y_c$ Position Error (mm) Comparison for \textbf{S1} and \textbf{S2}}
\label{table 1}
\centering
\begin{tabular}{|c||c||c||c||c||c|}
\hline
Delay (ms)& Controller & AE-$x_c$ & $\%$ AE-$x_c$. & AE-$y_c$ & $\%$ AE-$y_c$ \\
\hline
\multirow{2}{*}{\textbf{S1}}& PCON (\cite{Ref:32}) & 58.30 & 2.33 & 53.09 & 2.12\\ \cline{2-6}

& AROLC (proposed) & 23.33 & 0.93 & 22.92 & 0.92\\
\hline
\multirow{2}{*}{\textbf{S2}} & PCON (\cite{Ref:32}) & 78.46 & 3.14 & 81.37 & 3.25\\ \cline{2-6}

 & AROLC (proposed) & 36.42 & 1.46 & 39.66 & 1.59\\
\hline
\end{tabular}
\end{table}



\begin{figure}[t]
\begin{center}
\includegraphics[width=3 in,height=2.1 in]{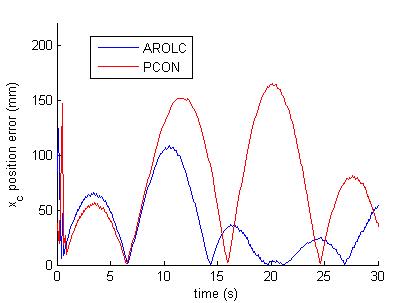}
\caption{$x_c$ position error comparison for input delay \textbf{S2}.}\label{fig:5}
\end{center}
\end{figure}

\begin{figure}[t]
\begin{center}
\includegraphics[width=3 in,height=2.1 in]{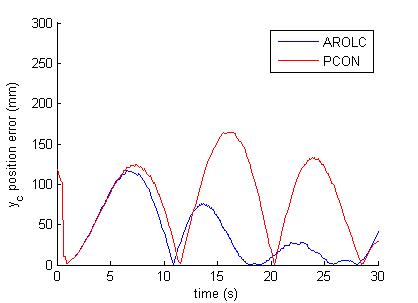}
\caption{$y_c$ position error comparison for input delay \textbf{S2}.}\label{fig:6}
\end{center}
\end{figure}

\par The performance of AROLC is also tested for the following two cases when the input delay is fixed and compared with the predictive controller reported in \cite{Ref:31} (denoted here as PCONf):
\par \textbf{S3}: $h(t)=60$ ms
\par \textbf{S4}: $h(t)=120$ ms
\par The comparative performance of AROLC and PCONf is provided in Table \ref{table 5} for \textbf{S3} and \textbf{S4}. Superior performance of the proposed controller over PCONf is clearly evident from the tabulated data. However, tracking accuracy of AROLC degrades as $h$ increases and this is commensurate with the fact that the error bands for AROLC increases with high input delay. The total variation (TV), a measure of smoothness of input, is denoted as \cite{Ref:39},
\begin{equation}
\text{TV}=\sum_{i=1}^{n-1}|u_r(i+1)-u_r(i))|+|u_l(i+1)-u_l(i))|
\end{equation}
where, $n$ is the length of the samples in $u_r$ and $u_l$ accumulated during experimentation. High value of TV denotes excessive usage of control input (\cite{Ref:39}). TV for the three controllers under various cases are provided in Table \ref{table 9}. It can be noticed that fixed time delay resulted in more control input requirement than the time varying delay for all the controllers. However, AROLC consumed the least control input for all the cases, which further augment its superior performance compared to PCON and PCONf. Some boxes in Table \ref{table 9} are left blank since those particular controllers were not used for the corresponding conditions. 


\begin{table}[!t]
\renewcommand{\arraystretch}{1.3}
\caption{$x_c$ and $y_c$ Position Error (mm) Comparison for \textbf{S3} and \textbf{S4}}
\label{table 5}
\centering
\begin{tabular}{|c||c||c||c||c||c|}
\hline
Delay (ms)& Controller & AE-$x_c$ & $\%$ AE-$x_c$ & AE-$y_c$ & $\%$ AE-$y_c$ \\
\hline
\multirow{2}{*}{\textbf{S3}}& PCONf (\cite{Ref:31}) & 99.63 & 3.99 & 85.79 & 3.43\\ \cline{2-6}

& AROLC (proposed) & 45.66 & 1.82 & 51.84 & 2.07\\
\hline
\multirow{2}{*}{\textbf{S4}}& PCONf (\cite{Ref:31}) & 136.91 & 5.48 & 102.19& 4.79\\ \cline{2-6}

&AROLC (proposed) & 82.29 & 3.29 & 69.05 & 2.76\\
\hline
\end{tabular}
\end{table}

\begin{table}[!t]
\renewcommand{\arraystretch}{1.3}
\caption{TV of Controllers for \textbf{S1}, \textbf{S2}, \textbf{S3} \textbf{S4}}
\label{table 9}
\centering
\begin{tabular}{|c||c||c||c||c|}
\hline
Controller & \textbf{S1} & \textbf{S2}. & \textbf{S3} & \textbf{S4} \\
\hline
PCONf (\cite{Ref:31}) & - & - & 8.74 & 10.73\\
\hline
PCON  (\cite{Ref:32}) & 5.45 & 6.39 & - & -\\
\hline
AROLC (proposed) & 0.42 & 0.55 & 2.23 & 3.36\\
\hline
\end{tabular}
\end{table}

\section{Conclusion}
In this paper, the tracking problem for a class of uncertain nonlinear systems in presence of unknown input delay is solved. The proposed delay dependent control law does not require any knowledge of the input delay to compute the control law and it is also insensitive towards the variation of the delay. The maximum sustainable delay is determined by utilizing the Razumikhin approach. AROLC, through its novel adaptive law, is further able to provide robustness against the parametric and unmodelled uncertainties without any prior knowledge of their bounds. Experimental results using a nonholonomic wheeled mobile robot, validates the superiority of the proposed controller compared to the predictor based control approach.


%

%
%
%


\bibliographystyle{spbasic}      
\end{document}